\documentclass[journal]{IEEEtran}
\usepackage{epsfig,graphics,epsf,color,amsmath,balance,cite,amsfonts,multirow,stfloats,algorithm,algorithmic,mathrsfs,amssymb}
\usepackage[dvipdfm,unicode,colorlinks, linkcolor=black,anchorcolor=black,citecolor=black]{hyperref}
\usepackage{mathtools}
\newtheorem{lemma}{Lemma}

\newtheorem{proposition}{Proposition}

\title{Secrecy Rate Maximization for Intelligent Reflecting Surface Assisted Multi-Antenna Communications}

\author{Hong Shen,~\IEEEmembership{Member,~IEEE,}~Wei Xu,~\IEEEmembership{Senior Member,~IEEE,}~Shulei Gong,\\ Zhenyao He,~and~Chunming Zhao,~\IEEEmembership{Member,~IEEE}\thanks{H. Shen, W. Xu, Z. He, and C. Zhao are with the National Mobile Communications Research Laboratory, Southeast University, Nanjing 210096, China (e-mail: \{shhseu, wxu, hezhenyao, cmzhao\}@seu.edu.cn). S. Gong is with the School of Electronic Science and Engineering, Nanjing University, Nanjing, 210093, China, and also with China Mobile Group Jiangsu Co., Ltd, Nanjing, 210029, China (email: gongshulei@smail.nju.edu.cn).}}

\begin{document}

\maketitle

\begin{abstract}
We investigate transmission optimization for intelligent reflecting surface (IRS) assisted multi-antenna systems from the physical-layer security perspective. The design goal is to maximize the system secrecy rate subject to the source transmit power constraint and the unit modulus constraints imposed on phase shifts at the IRS. To solve this complicated non-convex problem, we develop an efficient alternating algorithm where the solutions to the transmit covariance of the source and the phase shift matrix of the IRS are achieved in closed form and semi-closed forms, respectively. The convergence of the proposed algorithm is guaranteed theoretically.  Simulations results validate the performance advantage of the proposed optimized design.
\end{abstract}

\begin{keywords}
Intelligent reflecting surface (IRS), multi-antenna communications, physical-layer security, secrecy rate.
\end{keywords}

\section{Introduction}
Intelligent reflecting surface (IRS), which consists of a large number of low-cost passive reflecting elements with adjustable phase shifts, has recently been advocated as a cost-effective solution to significantly enhance the performance of wireless communications \cite{Tan2016ICC,Tan2018INFOCOM}. Owing to the tunable phase shifts of all the reflecting elements, the functions of signal enhancement and interference suppression can be achieved by the IRS without the use of active transmitters. \textcolor{black}{Compared to the well known massive multiple-input multiple-output (MIMO) technique, one appealing advantage of applying the IRS is to reduce the system energy consumption and achieve sustainable green 5G and beyond wireless networks.} \textcolor{black}{IRS can be implemented via conventional reflectarrays \cite{Tan2016ICC,Tan2018INFOCOM,Hum2014TAP}, liquid crystal metasurfaces \cite{Foo2017ISAP}, or software defined metamaterials \cite{Liaskos2018CM}. We highlight that the IRS is passive, operates under the full-duplex mode without self-interference and has no noise amplification, which is different from the relay system. Moreover, the IRS uses the passive reflecting elements for signal reflection, and thus differs from the active large intelligent surface (LIS) in \cite{Hu2017VTC,Hu2018TSP} which uses the entire surface for receiving and transmitting signals.}

Some innovative efforts have been devoted to system design and optimization for IRS-aided wireless communications \textcolor{black}{\cite{Wu2019ArxivCM,Wu2018GLOBECOM,Wu2018ArxivTWC,Huang2018ICASSP,Huang2018Arxiv,Wu2018Arxiv,Huang2018GLOBECOM}. Concretely, an overview on the IRS-aided wireless networks including the applications, hardware architecture, beamforming design, channel estimation, and network deployment was provided in \cite{Wu2019ArxivCM}.} The authors of \cite{Wu2018GLOBECOM} maximized the received signal-to-noise ratio (SNR) for a single-user multiple-input
single-output (MISO) system assisted by the IRS, where active transmit beamforming and passive reflected beamforming, i.e., the phase shifts of the IRS, were jointly optimized. In \cite{Wu2018ArxivTWC}, the SNR and signal-to-interference-plus-noise ratio (SINR) constrained transmit power minimization problems were investigated for single-user and multiuser IRS-aided MISO systems, respectively. \textcolor{black}{In particular, an interesting squared power gain regarding the number of reflecting elements was observed in \cite{Wu2018GLOBECOM,Wu2018ArxivTWC}.} Alternatively, the authors of \cite{Huang2018ICASSP} and \cite{Huang2018Arxiv} considered maximizing the sum rate and energy efficiency of a multiuser MISO system, respectively, by jointly optimizing the transmit powers for all users and the phase shifts of the IRS. Different from the above works that focused on continuous phase shifts at the IRS, discrete phase shifts were further concerned for both single-user  \cite{Wu2018Arxiv} and multiuser \cite{Huang2018GLOBECOM} MISO systems. To summarize, IRS has been shown to be beneficial for improving the performance of multi-antenna systems in terms of, e.g., achievable rate and energy efficiency. Nevertheless, to the best of our knowledge, \textcolor{black}{it has rarely been investigated in prior works to optimize the secrecy performance of IRS assisted multi-antenna systems, which is a vital issue when the transmitted signal is subject to interception.}

In this letter, following the philosophy of the physical-layer security \cite{Khisti2010TIT,Li2011TSP}, we aim to maximize the secrecy rate of an IRS-aided multi-antenna system  by jointly optimizing the source transmit covariance and IRS's phase shift matrix.  \textcolor{black}{The problem is quite challenging even when the eavesdropper has a single antenna, which cannot be straightforwardly solved with the techniques developed in \cite{Wu2018GLOBECOM,Wu2018ArxivTWC,Huang2018ICASSP,Huang2018Arxiv,Wu2018Arxiv,Huang2018GLOBECOM}.} To acquire a tractable solution, we propose an alternating algorithm that optimizes one variable with the other fixed. Specifically, we first obtain a closed-form solution to the source transmit covariance. Then, concerning the difficult phase shift matrix optimization, we develop a bisection search based semi-closed form solution via tight bounding. The convergence of the alternating algorithm is proved rigorously. We further extend the proposed algorithm to the general case where the eavesdropper has multiple antennas.

\emph{Notations:} Vectors and matrices are denoted by boldface lower-case and boldface upper-case letters, respectively. $(\cdot)^*$, $(\cdot)^{T}$, and $(\cdot)^{H}$ stand for the conjugate, the transpose, and the Hermitian operations, respectively. $|\cdot|$ and $\|\cdot\|$ denote the absolute value of a scalar and the $\ell_2$ norm of a vector, respectively.  $\text{diag}\{\cdot\}$ represents the diagonal matrix whose diagonals are the elements of the input vector. $\text{tr}(\cdot)$, $\det(\cdot)$, and $\lambda_\text{max}(\cdot)$ return the trace, the determinant, and the maximum eigenvalue of the input matrix, respectively. $\mathbb{E}\{\cdot\}$ is the expectation operation. $\Re(\cdot)$ and $\text{arg}(\cdot)$ return the real part and the phase of the input complex number, respectively. By $\mathbf A \succeq \mathbf 0$, we mean that matrix $\mathbf A$ is positive semidefinite.

\section{System Model and Problem Formulation}\label{sec:model}
\subsection{System Model Description}
Consider an IRS assisted multi-antenna system including one source (Alice), one IRS, one legitimate receiver (Bob), and one eavesdropper (Eve)\textcolor{black}{, as depicted in Fig.~\ref{fig:model}}. Alice has $N$ antennas, IRS has $L$  low-cost passive reflecting elements, and both Bob and Eve are single-antenna nodes. \textcolor{black}{Note that we assume that the power of the signals reflected by the IRS two or more times is quite small and thus neglected. Moreover, we consider maximal reflection without loss at the IRS since each reflecting element of the IRS should be designed to maximize the power of the reflected signal.}

When Alice transmits a secret message to Bob, Bob receives the signals from both Alice and the IRS since the IRS reflects the signals from Alice. Accordingly, the received signal at Bob is ${y}_B\!=\!\mathbf{h}_{IB}^H\mathbf{\Theta}\mathbf{H}_{AI}\mathbf x\!+\!
\mathbf{h}_{AB}^H\mathbf x\!+\!{n}_B$
where $\mathbf{h}_{IB}^H$ is the IRS-to-Bob channel, $\mathbf{\Theta}=\text{diag}\{[e^{j\phi_1},\cdots,e^{j\phi_i},\cdots,e^{j\phi_{L}}]\}$ is a diagonal matrix with $\phi_i$ denoting the phase shift incurred by the $i$-th reflecting element of the IRS, $\mathbf{H}_{AI}$ is the Alice-to-IRS channel, $\mathbf x$ is Alice's transmit signal whose covariance matrix $\mathbf W \!\triangleq \!\mathbb{E}\{\mathbf x\mathbf x^H\}$ satisfies $ \text{tr}(\mathbf W) \!\leq \!P$ with $P$ denoting the maximum transmit power, $\mathbf{h}_{AB}^H$ is the Alice-to-Bob channel, and ${n}_B$ is the additive white Gaussian noise (AWGN) at Bob with variance $\sigma_{n,B}^2$. Similarly, the signal received by Eve is ${y}_E\!=\!\mathbf{h}_{IE}^H\mathbf{\Theta}\mathbf{H}_{AI}\mathbf x\!+\!
\mathbf{h}_{AE}^H\mathbf x\!+\!{n}_E$
where $\mathbf{h}_{IE}^H$, $\mathbf{h}_{AE}^H$, and ${n}_E$ are the IRS-to-Eve channel, the Alice-to-Eve channel, and the AWGN at Eve with variance $\sigma_{n,E}^2$, respectively.

\begin{figure}[t]
    \begin{center}
      \epsfxsize=7.0in\includegraphics[scale=0.4]{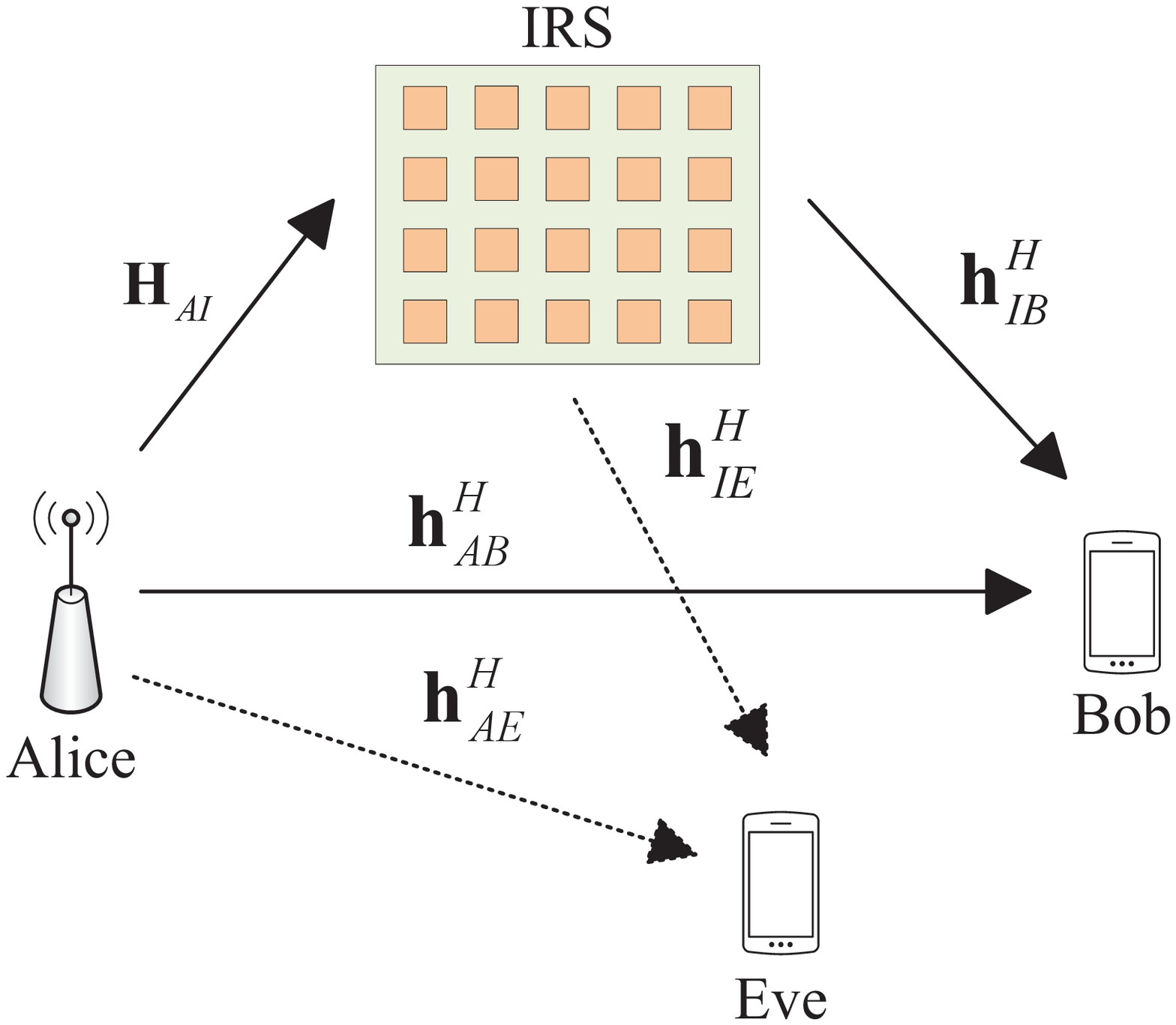}
      \caption{An IRS-aided wireless communication system subject to interception.}\label{fig:model}
    \end{center}
\end{figure}
\subsection{Secrecy Rate Maximization Problem}
To enhance the security of the above system from the physical layer perspective, we jointly optimize $\mathbf{W}$ and $\mathbf{\Theta}$ such that the system secrecy rate is maximized. Specifically, the achievable secrecy rate is given by
\begin{align}
& R_s(\mathbf{W},\mathbf{\Theta})=\nonumber \\
&\log_2\left(1\!+\!\frac{(\mathbf{h}_{IB}^H\mathbf{\Theta}\mathbf{H}_{AI}+
\mathbf{h}_{AB}^H)\mathbf W(\mathbf{H}_{AI}^H\mathbf{\Theta}^H\mathbf{h}_{IB}+
\mathbf{h}_{AB})}{\sigma_{n,B}^2}\right)\!\!-\nonumber \\
&\log_2\left(\!\!1+\!\!\frac{(\mathbf{h}_{IE}^H\mathbf{\Theta}\mathbf{H}_{AI}\!\!+\!\!
\mathbf{h}_{AE}^H)\mathbf W(\mathbf{H}_{AI}^H\mathbf{\Theta}^H\mathbf{h}_{IE}\!\!+\!\!
\mathbf{h}_{AE})}{\sigma_{n,E}^2}\!\!\right). \label{eq:secrate}
\end{align}
Furthermore, we need to impose a power constraint on $\mathbf W$ and unit modulus constraints on the diagonals of $\mathbf{\Theta}$. Accordingly, we formulate the problem of interest as
\begin{align}\label{eq:SecRateMaxProb}
\mathop{\text{maximize}}\limits_{\mathbf{W} \succeq \mathbf 0,\mathbf{\Theta}} \quad &
R_s(\mathbf{W},\mathbf{\Theta})\nonumber\\
\text{subject to}\quad &  \text{tr}(\mathbf W) \leq P,\ |\theta_i|=1,\ i=1,\cdots,L,
\end{align}
where $\theta_i$ is the $i$-th diagonal of $\mathbf{\Theta}$.
It is non-trivial to solve this problem since the optimization variables are coupled in the objective function and there exist unit modulus constraints which are usually hard to handle.

\section{Alternating Algorithm for Secrecy Rate Maximization Problem}\label{sec:SecRateSingle}
In this section, we develop an efficient algorithm for problem \eqref{eq:SecRateMaxProb}  which optimizes $\mathbf W$ and $\mathbf{\Theta}$ in an alternating manner. In particular, the optimal solution to $\mathbf W$ with fixed $\mathbf{\Theta}$ has a closed form, while the optimization of  $\mathbf{\Theta}$ with given $\mathbf W$ admits a semi-closed form solution.
\subsection{Closed-Form Solution to $\mathbf W$ With Given $\mathbf{\Theta}$}
By fixing $\mathbf{\Theta}$, the optimization with respect to $\mathbf W$ becomes:
\begin{align}\label{eq:SecRateMaxProbwa}
\mathop{\text{maximize}}\limits_{\mathbf{W} \succeq \mathbf 0} \quad & \frac{\mathbf{h}_B^H\mathbf W\mathbf{h}_B+\sigma_{n,B}^2}{\mathbf{h}_E^H\mathbf W\mathbf{h}_E+\sigma_{n,E}^2} \
\text{subject to}\quad \text{tr}(\mathbf W) \leq P,
\end{align}
where $\mathbf{h}_B=\mathbf{H}_{AI}^H\mathbf{\Theta}^H\mathbf{h}_{IB}+
\mathbf{h}_{AB}$ and $\mathbf{h}_E=\mathbf{H}_{AI}^H\mathbf{\Theta}^H\mathbf{h}_{IE}+
\mathbf{h}_{AE}$. According to \cite{Khisti2010TIT,Li2011TSP}, the optimal solution to $\mathbf{W}$ is
\begin{align}\label{eq:WAopt}
\mathbf{W}^{\star}={P}\mathbf {\tilde w}\mathbf {\tilde w}^H,
\end{align}
where $\mathbf {\tilde w}$ is the normalized dominant generalized eigenvector of the matrix pencil
$(P\mathbf{h}_B\mathbf{h}_B^H+{\sigma_{n,B}^2}\mathbf I,P \mathbf{h}_E\mathbf{h}_E^H+{\sigma_{n,E}^2}\mathbf I)$.

\subsection{Optimization of $\mathbf \Theta$ With Given $\mathbf W$}
Now let us perform optimization over $\mathbf \Theta$ by fixing $\mathbf W$. The corresponding problem is
\begin{align}\label{eq:SecRateMaxProbtheta}
\mathop{\text{maximize}}\limits_{\mathbf{\Theta}} \quad &\frac{|(\mathbf{h}_{IB}^H\mathbf{\Theta}\mathbf{H}_{AI}+
\mathbf{h}_{AB}^H)\mathbf w|^2+\sigma_{n,B}^2}{|(\mathbf{h}_{IE}^H\mathbf{\Theta}\mathbf{H}_{AI}+
\mathbf{h}_{AE}^H)\mathbf w|^2+\sigma_{n,E}^2} \nonumber\\
\text{subject to}\quad & |\theta_i|=1,\ i=1,\cdots,L,
\end{align}
where we define $\mathbf {w}=\sqrt{P}\mathbf {\tilde w}$.
Different from the optimization of $\mathbf W$, it is quite hard to achieve the optimal solution to this problem due to the unit modulus constraints.

To solve problem \eqref{eq:SecRateMaxProbtheta}, we first define $\boldsymbol \theta=[\theta_1^*,\cdots,\theta_{L}^*]^T$. Then, by invoking the equality $\mathbf a^H \mathbf{\Theta} \mathbf b=\boldsymbol \theta^H\text{diag}\{\mathbf a^H\} \mathbf b$, we rewrite problem \eqref{eq:SecRateMaxProbtheta} by
\begin{align}\label{eq:SecRateMaxProbtheta1}
\mathop{\text{maximize}}\limits_{\boldsymbol{\theta}} \quad &\frac{|\boldsymbol \theta^H \boldsymbol \alpha_B+\tilde \alpha_B |^2+\sigma_{n,B}^2}{|\boldsymbol \theta^H \boldsymbol \alpha_E+\tilde \alpha_E|^2+\sigma_{n,E}^2} \nonumber\\
\text{subject to}\quad & |\theta_i|=1,\ i=1,\cdots,L,
\end{align}
where $\boldsymbol \alpha_B\!\!=\!\!\text{diag}\{\mathbf{h}_{IB}^H\}\mathbf{H}_{AI}\mathbf w$, $\tilde \alpha_B\!\!=\!\!\mathbf{h}_{AB}^H\mathbf w$, $\boldsymbol \alpha_E\!\!=\!\!\text{diag}\{\mathbf{h}_{IE}^H\}\mathbf{H}_{AI}\mathbf w$, and $\tilde \alpha_E\!=\!\mathbf{h}_{AE}^H\mathbf w$. Rewrite problem \eqref{eq:SecRateMaxProbtheta1} by
\begin{align}\label{eq:SecRateMaxProbtheta2}
\mathop{\text{minimize}}\limits_{\boldsymbol{\theta}} \quad &\frac{|\boldsymbol \theta^H \boldsymbol \alpha_E+\tilde \alpha_E|^2+\sigma_{n,E}^2}{|\boldsymbol \theta^H \boldsymbol \alpha_B+\tilde \alpha_B |^2+\sigma_{n,B}^2} \nonumber\\
\text{subject to}\quad & |\theta_i|=1,\ i=1,\cdots,L.
\end{align}
This problem belongs to fractional programming. Following \cite{Dinkelbach1967}, we consider the corresponding parametric program:
\begin{align}\label{eq:SecRateMaxProbtheta3}
\mathop{\text{minimize}}\limits_{\boldsymbol \theta} \quad &|\boldsymbol \theta^H \boldsymbol \alpha_E\!+\!\tilde \alpha_E|^2\!+\!\sigma_{n,E}^2\!-\!\mu(|\boldsymbol \theta^H \boldsymbol \alpha_B\!+\!\tilde \alpha_B |^2\!+\!\sigma_{n,B}^2) \nonumber\\
\text{subject to}\quad & |\theta_i|=1,\ i=1,\cdots,L,
\end{align}
where $\mu \geq 0$ is an introduced parameter. Denote the optimal objective value of this problem by $\psi^\star(\mu)$. Then, the optimal objective value of problem \eqref{eq:SecRateMaxProbtheta2} is the unique root of  $\psi^\star(\mu)=0$ \cite{Dinkelbach1967}. Finding the root requires solving problem \eqref{eq:SecRateMaxProbtheta3} with given $\mu$, which is, however, still non-convex and hard to be solved. To make it more tractable, we minimize an upper bound of its objective function, which is given by \cite[Example 13]{Sun2017TSP}
\begin{align}\label{eq:objupper}
& {|\boldsymbol \theta^H \boldsymbol \alpha_E+\tilde \alpha_E|^2+\sigma_{n,E}^2}-\mu(|\boldsymbol \theta^H \boldsymbol \alpha_B+\tilde \alpha_B |^2+\sigma_{n,B}^2) \nonumber\\
& = \!\!\boldsymbol \theta^H (\boldsymbol \alpha_E\boldsymbol \alpha_E^H\!\!-\!\!\mu\boldsymbol \alpha_B\boldsymbol \alpha_B^H) \boldsymbol \theta\!\!-\!\!2\Re\{\boldsymbol \theta^H( \mu\tilde \alpha_B^*\boldsymbol \alpha_B\!\!-\!\!\tilde \alpha_E^*\boldsymbol \alpha_E)\}\!\!+\!\!|\tilde \alpha_E|^2\nonumber\\
& +\!\!\sigma_{n,E}^2\!\!-\!\!\mu|\tilde \alpha_B|^2\!\!-\!\!\mu\sigma_{n,B}^2 \!\!\leq \!\!\lambda_{\text{max}}(\mathbf \Phi)\|\boldsymbol \theta\|^2\!\!-\!\!2\Re\{\boldsymbol \theta^H\boldsymbol \beta\}\!\!+\!\!c,
\end{align}
where $\mathbf \Phi=\boldsymbol \alpha_E\boldsymbol \alpha_E^H-\mu\boldsymbol \alpha_B\boldsymbol \alpha_B^H$, $\boldsymbol \beta=(\lambda_{\text{max}}(\mathbf \Phi)\mathbf I-\mathbf \Phi)\boldsymbol {\tilde \theta}+\mu\tilde \alpha_B^*\boldsymbol \alpha_B-\tilde \alpha_E^*\boldsymbol \alpha_E$, $c=\boldsymbol {\tilde \theta}^H(\lambda_{\text{max}}(\mathbf \Phi)\mathbf I-\mathbf \Phi) \boldsymbol {\tilde \theta}+|\tilde \alpha_E|^2+\sigma_{n,E}^2-\mu|\tilde \alpha_B|^2-\mu\sigma_{n,B}^2$, and $\boldsymbol {\tilde \theta}$ is the solution to $\boldsymbol \theta$ obtained in the previous iteration of the alternating algorithm.
The simplified optimization problem becomes
\begin{align}\label{eq:SecRateMaxProbtheta4}
\mathop{\text{minimize}}\limits_{\boldsymbol \theta} \quad &\lambda_{\text{max}}(\mathbf \Phi)\|\boldsymbol \theta\|^2-2\Re\{\boldsymbol \theta^H\boldsymbol \beta\} \nonumber\\
\text{subject to}\quad & |\theta_i|=1,\ i=1,\cdots,L.
\end{align}
Clearly, $\|\boldsymbol \theta\|^2=L$ since $|\theta_i|=1$. Moreover, $\Re\{\boldsymbol \theta^H\boldsymbol \beta\}$ is maximized when the phases of $\theta_i$ and $\beta_i$ are equal, where $\beta_i$ is the $i$-th entry of $\boldsymbol \beta$. Therefore, the optimal solution to problem \eqref{eq:SecRateMaxProbtheta4} with given $\mu$ is
\begin{align}\label{eq:thetaopt}
\boldsymbol \theta^{\star}(\mu)=[e^{j\text{arg}(\beta_1)},\cdots,e^{j\text{arg}(\beta_{L})}]^T.
\end{align}
\textcolor{black}{Although problem \eqref{eq:SecRateMaxProbtheta3} can also be handled by the semidefinite relaxation (SDR) \cite{Luo2010SPM}, the above solution has a closed form which is more convenient for implementation and requires much lower complexity especially for large $L$.} Substitute $\boldsymbol \theta^{\star}(\mu)$ into the objective function of problem \eqref{eq:SecRateMaxProbtheta3} and denote the result by $\tilde \psi^\star(\mu)$. Then, the following lemma holds.
\begin{lemma}\label{lemma:psimu}
$\tilde \psi^\star(\mu)$ is a strictly decreasing function.
\end{lemma}
\begin{proof}
See Appendix~\ref{app:lemma1}.
\end{proof}
Based on \emph{Lemma}~\ref{lemma:psimu} and the facts that $\tilde \psi^\star(0)>0$ and $\tilde \psi^\star(+\infty)<0$, we conclude that $\tilde \psi^\star(\mu)=0$ has a unique root (denoted by \textcolor{black}{$\mu^{'}$}), which can be determined via bisection search. Then, we obtain the solution to $\boldsymbol \theta$ by $\boldsymbol \theta^{\star}(\textcolor{black}{\mu^{'}})$.

We summarize the proposed algorithm that solves problem \eqref{eq:SecRateMaxProb} in Algorithm~\ref{alg1}. \textcolor{black}{Note that problem \eqref{eq:SecRateMaxProb} is always feasible and Algorithm~\ref{alg1} yields a feasible solution.} Furthermore, we can prove that Algorithm~\ref{alg1} must converge based on the following proposition.
\begin{proposition}\label{prop:larger}
The objective value of problem \eqref{eq:SecRateMaxProbtheta1} in the current iteration is no smaller than that in the previous iteration, i.e., $\gamma(\boldsymbol \theta^{\star}(\textcolor{black}{\mu^{'}})) \geq \gamma(\boldsymbol {\tilde \theta})$, where $\gamma(\boldsymbol \theta)$ denotes the objective function of problem \eqref{eq:SecRateMaxProbtheta1}.
\end{proposition}
\begin{proof}
See Appendix~\ref{app:prop1}.
\end{proof}
Owing to the above conclusion and the optimality of $\mathbf W^{\star}$, the objective value of problem \eqref{eq:SecRateMaxProb} is non-decreasing after each iteration of Algorithm~\ref{alg1}. Moreover, the objective value has a finite upper bound. Therefore, Algorithm~\ref{alg1} always converges. \textcolor{black}{Note that transforming problem \eqref{eq:SecRateMaxProbtheta1} into problem \eqref{eq:SecRateMaxProbtheta2} is necessary since otherwise the alternating algorithm does not converge as verified via numerical tests.}

\begin{algorithm}[t]
\caption{\textbf{Alternating algorithm for problem \eqref{eq:SecRateMaxProb}}}\label{alg1}

\begin{algorithmic}[1]

\STATE \textit{Initialization:}  set initial $\boldsymbol {\tilde \Theta}$ and convergence accuracy $\epsilon$.
\STATE \textbf{repeat}

\STATE \text{  } Fix $\boldsymbol \Theta=\boldsymbol {\tilde \Theta}$ and calculate $\mathbf W^{\star}$ using \eqref{eq:WAopt}.

\STATE \text{  } Fix $\mathbf W=\mathbf W^{\star}$ and find the root of $\tilde \psi^\star(\mu)=0$, i.e., \textcolor{black}{$\mu^{'}$}, using bisection search and \eqref{eq:thetaopt}.

\STATE \text{  } Set $\boldsymbol {\tilde \Theta}=\text{diag}\{\boldsymbol { \theta}^{\star}(\textcolor{black}{\mu^{'}})\}$.

\STATE \textbf{until} convergence.
\end{algorithmic}
\end{algorithm}
\section{Extension for Multi-Antenna Eve}\label{sec:SecRateMulti}
The proposed algorithm can also be extended to address the general case where Eve has $M$ antennas ($M>1$). For this scenario, the secrecy rate becomes 
\begin{align}
&R_s(\mathbf{W},\!\mathbf{\Theta})\!\!=  \nonumber \\ &\log_2\!\left(\!\!1\!\!+\!\!\frac{(\mathbf{h}_{IB}^H\mathbf{\Theta}\mathbf{H}_{AI}\!\!+\!\!
\mathbf{h}_{AB}^H)\mathbf W(\mathbf{H}_{AI}^H\mathbf{\Theta}^H\mathbf{h}_{IB}\!\!+\!\!
\mathbf{h}_{AB})}{\sigma_{n,B}^2}\!\!\right)\!\!-\nonumber \\
& \log_2\!\det\!\left(\!\!\mathbf I\!\!+\!\!\frac{(\mathbf{H}_{IE}\mathbf{\Theta}\mathbf{H}_{AI}\!\!+\!\!
\mathbf{H}_{AE})\mathbf W(\mathbf{H}_{AI}^H\mathbf{\Theta}^H\mathbf{H}_{IE}^H\!\!+\!\!
\mathbf{H}_{AE}^H)}{\sigma_{n,E}^2}\!\!\right), \label{eq:secrateMul}
\end{align}
where $\mathbf{H}_{IE}=[\mathbf{h}_{IE,1},\cdots,\mathbf{h}_{IE,M}]^H$ and $\mathbf{H}_{AE}=[\mathbf{h}_{AE,1},\cdots,\mathbf{h}_{AE,M}]^H$ denote the IRS-to-Eve channel and the Alice-to-Eve channel, respectively.

By fixing $\mathbf{\Theta}$, the optimal solution to $\mathbf W$ is  \cite{Khisti2010TIT,Li2011TSP}
\setcounter{equation}{12}
\begin{align}\label{eq:WAoptmul}
\mathbf{W}^{\star}={P}\mathbf {\tilde w}\mathbf {\tilde w}^H,
\end{align}
where $\mathbf {\tilde w}$ is the normalized dominant generalized eigenvector of the matrix pencil
$(P\mathbf{h}_B\mathbf{h}_B^H+{\sigma_{n,B}^2}\mathbf I,P \mathbf{H}_E\mathbf{H}_E^H+{\sigma_{n,E}^2}\mathbf I)$ with $\mathbf{H}_E=\mathbf{H}_{AI}^H\mathbf{\Theta}^H\mathbf{H}_{IE}^H+
\mathbf{H}_{AE}^H$. On the other hand, based on \eqref{eq:WAoptmul} and $\det(\mathbf I+\mathbf {AB})=\det(\mathbf I+\mathbf {BA})$,  the optimization of $\mathbf \Theta$ with fixed $\mathbf W$ can be expressed by
\begin{align}\label{eq:SecRateMaxProbthetaMul}
\mathop{\text{maximize}}\limits_{\mathbf{\Theta}} \quad &\frac{|(\mathbf{h}_{IB}^H\mathbf{\Theta}\mathbf{H}_{AI}+
\mathbf{h}_{AB}^H)\mathbf w|^2+\sigma_{n,B}^2}{\sum_{i=1}^{M}|(\mathbf{h}_{IE,i}^H\mathbf{\Theta}\mathbf{H}_{AI}+
\mathbf{h}_{AE,i}^H)\mathbf w|^2+\sigma_{n,E}^2} \nonumber\\
\text{subject to}\quad & |\theta_i|=1,\ i=1,\cdots,L.
\end{align}
The method of solving problem \eqref{eq:SecRateMaxProbtheta}  can be adopted with some modifications to address this problem.  The differences are as follows. Firstly, when calculating $\boldsymbol {\theta}^{\star}(\mu)$ using \eqref{eq:thetaopt},  $\boldsymbol \beta$ now becomes $(\lambda_{\text{max}}(\mathbf \Phi)\mathbf I-\mathbf \Phi)\boldsymbol {\tilde \theta}+\mu\tilde \alpha_B^*\boldsymbol \alpha_B-\sum_{i=1}^{M}\tilde \alpha_{E,i}^*\boldsymbol \alpha_{E,i}$ where $\mathbf \Phi=\sum_{i=1}^{M}\boldsymbol \alpha_{E,i}\boldsymbol \alpha_{E,i}^H-\mu\boldsymbol \alpha_B\boldsymbol \alpha_B^H$, $\tilde \alpha_{E,i}=\mathbf{h}_{AE,i}^H\mathbf w$, and $\boldsymbol \alpha_{E,i}=\text{diag}\{\mathbf{h}_{IE,i}^H\}\mathbf{H}_{AI}\mathbf w$. Secondly, the function $\tilde \psi^\star(\mu)$ is updated by ${\sum_{i=1}^{M}|(\boldsymbol {\theta}^{\star}(\mu))^H \boldsymbol \alpha_{E,i}+\tilde \alpha_{E,i}|^2+\sigma_{n,E}^2}-\mu(|(\boldsymbol {\theta}^{\star}(\mu))^H \boldsymbol \alpha_B+\tilde \alpha_B |^2+\sigma_{n,B}^2)$, which can also be verified to be a strictly decreasing function.

\section{Simulation Results}
The performance of the proposed secrecy rate maximized design is evaluated via simulations for an IRS assisted multi-antenna system. Moreover, we also consider the case without IRS as a benchmark, where we optimize the transmit covariance of Alice to maximize the secrecy rate \cite{Khisti2010TIT,Li2011TSP}. \textcolor{black}{We set $N=4$, $P_A=15\ \text{dBW}$, and $\sigma_{n,B}^2=\sigma_{n,E}^2=-75\ \text{dBW}$. The small-scale fading of all the channels follows the Rayleigh fading model. The path loss model is given by $\text{PL}=\left(\text{PL}_0-10\zeta\log_{10}\left(\frac{d}{d_0}\right)\right)\ \text{dB}$,
where $\text{PL}_0$ is the path loss at the reference distance $d_0$, $\zeta$ is the path loss exponent, and $d$ is the distance between the transmitter and the receiver. We set $\text{PL}_0=-30\text{ dB}$ and $d_0=1 \text{m}$. The path loss exponents of the Alice-to-IRS link, the IRS-to-Bob link, the IRS-to-Eve link, the Alice-to-Bob link, and the Alice-to-Eve link are set to $\zeta_{AI}=2.2$, $\zeta_{IB}=2.5$, $\zeta_{IE}=2.5$, $\zeta_{AB}=3.5$, and $\zeta_{AE}=3.5$, respectively. The distance between Alice and the IRS is $d_{AI}=50\ \text{m}$. Both Bob and Eve lie in a horizontal line which is parallel to the one between Alice and the IRS. The vertical distance between these two lines is $d_v=2\ \text{m}$. The horizontal distance between Alice and Eve is $d_{AE,h}=44\ \text{m}$.}

\begin{figure}[t]
    \begin{center}
      \epsfxsize=7.0in\includegraphics[scale=0.5]{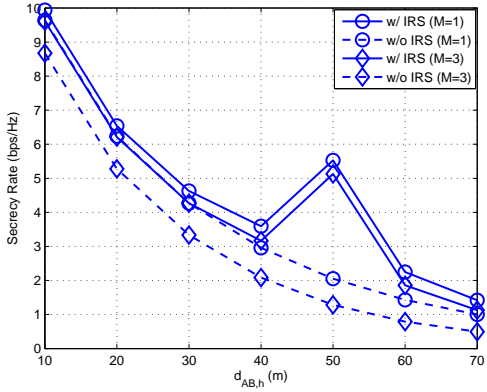}
      \caption{\textcolor{black}{Secrecy rate versus the horizontal distance between Alice and Bob.}}\label{fig:dabh}
    \end{center}
\end{figure}


\begin{figure}[t]
    \begin{center}
      \epsfxsize=7.0in\includegraphics[scale=0.5]{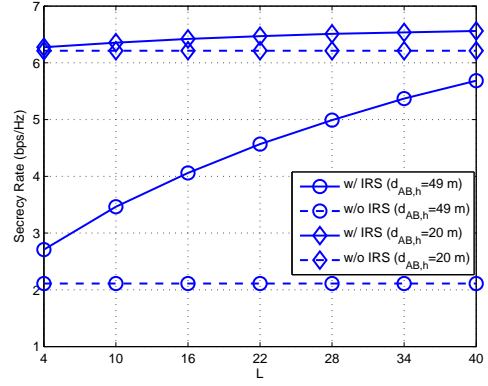}
      \caption{\textcolor{black}{Secrecy rate versus the number of reflecting elements.}}\label{fig:NumReflect}
    \end{center}
\end{figure}
\textcolor{black}{We show the secrecy rate performance by varying $d_{AB,h}$ in Fig.~\ref{fig:dabh}, where $L=32$. From this figure, we observe that the proposed IRS assisted design provides a higher secrecy rate than the conventional scheme without IRS. In particular, when there is no IRS, the secrecy rate of the conventional method gradually decreases with the increase of $d_{AB,h}$ as expected. While for the proposed method, the secrecy rate becomes increasing with respect to $d_{AB,h}$ when $d_{AB,h}\in[40\ \text{m},50\ \text{m}]$. This is due to the fact that the IRS can effectively enhance Bob's achievable rate via reflect beamforming when Bob is close to the IRS. We also find from Fig.~\ref{fig:dabh} that the secrecy rate of the proposed method degrades when $M$ becomes larger. This is because the achievable rate of Eve becomes higher with more antennas. The secrecy rate versus $L$ is shown in Fig.~\ref{fig:NumReflect}, where $M=1$. It can be seen that, when Bob is far from the IRS, the performance of the proposed IRS aided design is not quite sensitive to the value of $L$ since the signal from the IRS is weak at Bob even for large $L$.  When Bob is close to the IRS,  the secrecy rate achieved by the proposed method increases significantly as $L$ gets larger because the signal from the IRS becomes dominant at Bob. \textcolor{black}{Note that the above conclusions also hold after replacing the x-axis of Fig.~\ref{fig:NumReflect} by the area of IRS when the distance between the adjacent reflecting elements is fixed.} During the simulation, the squared power gain revealed in \cite{Wu2018GLOBECOM,Wu2018ArxivTWC} is not observed at Bob or Eve since we aim to maximize the secrecy rate instead of Bob's rate only. Finally, when $L$ is relatively small, we also find via simulations that the proposed method can achieve almost the same secrecy rate as the grid search based optimal solution.}

\section{Conclusions}
We studied the secrecy rate maximization for an IRS assisted multi-antenna system, where Alice's transmit covariance and IRS's phase shift matrix were jointly optimized. We advocated an efficient algorithm to optimize the two variables in an alternating manner. Closed-form and semi-closed form solutions were successfully obtained for the transmit covariance of Alice and the phase shift matrix of the IRS, respectively. The superiority of the proposed design has been confirmed via simulations.

\begin{appendices}
\section{Proof of \emph{Lemma}~\ref{lemma:psimu}}\label{app:lemma1}
Denote the left-hand side and right-hand side of the inequality in \eqref{eq:objupper} by $f(\boldsymbol \theta|\mu)$ and $g(\boldsymbol \theta|(\mu,\boldsymbol {\tilde \theta}))$, respectively. Suppose that $0 < \mu_1<\mu_2$. Then, we have $\tilde \psi^\star(\mu_2)=f(\boldsymbol \theta^{\star}(\mu_2)|\mu_2) \overset{(a)}{\leq} g((\boldsymbol \theta^{\star}(\mu_2))|(\mu_2,\boldsymbol \theta^{\star}(\mu_1))) \overset{(b)}{\leq}  g((\boldsymbol \theta^{\star}(\mu_1))|(\mu_2,\boldsymbol \theta^{\star}(\mu_1))) \overset{(c)}{=} f(\boldsymbol \theta^{\star}(\mu_1)|\mu_2)\overset{(d)}{<} f(\boldsymbol \theta^{\star}(\mu_1)|\mu_1)=\tilde \psi^\star(\mu_1)$, where (a) follows from \eqref{eq:objupper} with $\boldsymbol {\tilde \theta}=\boldsymbol \theta^{\star}(\mu_1)$, (b) holds because $\boldsymbol \theta^{\star}(\mu_2)$ minimizes $g(\boldsymbol \theta|(\mu_2,\boldsymbol {\tilde \theta}))$, (c) holds due to the equality $f(\boldsymbol \theta|\mu)=g(\boldsymbol \theta|(\mu,\boldsymbol {\theta}))$, and (d) holds because we assume that $\mu_1<\mu_2$.

\section{Proof of \emph{Proposition}~\ref{prop:larger}}\label{app:prop1}
Denote $|\boldsymbol \theta^H \boldsymbol \alpha_E\!+\!\tilde \alpha_E|^2\!+\!\sigma_{n,E}^2$ and $|\boldsymbol \theta^H \boldsymbol \alpha_B\!+\!\tilde \alpha_B |^2\!+\!\sigma_{n,B}^2$  by $f_E(\boldsymbol \theta)$ and $f_B(\boldsymbol \theta)$, respectively. Then, it follows that $f_E(\boldsymbol {\tilde \theta})-\textcolor{black}{\mu^{'}} f_B(\boldsymbol {\tilde \theta})=f(\boldsymbol {\tilde \theta}|\textcolor{black}{\mu^{'}})\overset{(a)}{=}g(\boldsymbol {\tilde \theta}|(\textcolor{black}{\mu^{'}},\boldsymbol {\tilde \theta})) \overset{(b)}{\geq} g(\boldsymbol {\theta}^{\star}(\textcolor{black}{\mu^{'}})|(\textcolor{black}{\mu^{'}},\boldsymbol {\tilde \theta})) \overset{(c)}{\geq} f(\boldsymbol {\theta}^{\star}(\textcolor{black}{\mu^{'}})|\textcolor{black}{\mu^{'}})= f_E(\boldsymbol {\theta}^{\star}(\textcolor{black}{\mu^{'}}))-\textcolor{black}{\mu^{'}} f_B(\boldsymbol {\theta}^{\star}(\textcolor{black}{\mu^{'}}))\overset{\textcolor{black}{(d)}}{=}0$, where (a) follows from $f(\boldsymbol \theta|\mu)=g(\boldsymbol \theta|(\mu,\boldsymbol {\theta}))$, (b) holds since $\boldsymbol \theta^{\star}(\textcolor{black}{\mu^{'}})$ minimizes $g(\boldsymbol \theta|(\textcolor{black}{\mu^{'}},\boldsymbol {\tilde \theta}))$, (c) is due to \eqref{eq:objupper}, \textcolor{black}{and (d) holds because $\textcolor{black}{\mu^{'}}$ is the unique root of $\tilde \psi^{\star}(\mu)=0$}. Therefore, we have $\gamma(\boldsymbol {\tilde \theta})=\frac{f_B(\boldsymbol {\tilde \theta})}{f_E(\boldsymbol {\tilde \theta})} \leq \frac{1}{\textcolor{black}{\mu^{'}}} =\frac{f_B(\boldsymbol {\theta}^{\star}(\textcolor{black}{\mu^{'}}))}{f_E(\boldsymbol {\theta}^{\star}(\textcolor{black}{\mu^{'}}))}=\gamma(\boldsymbol \theta^{\star}(\textcolor{black}{\mu^{'}}))$.

\end{appendices}

\end{document}